\documentclass[12pt]{article}
\usepackage{times,parskip,epsf,epsfig,latexsym,rotating,natbib}
\usepackage{isomath,caption,subcaption}
\usepackage{graphicx}
\usepackage{tikz}
\usepackage{url}
\usepackage{enumerate}
\usepackage{amsmath}
\usepackage{amsthm}
\usepackage{amssymb}
\usepackage{multirow}
\usepackage{comment}
\usepackage{amsfonts}
\usepackage{rotating}
\usepackage{graphicx}
\usepackage{lipsum}
\usepackage{hyperref}
\usepackage{bookmark,booktabs}

\PassOptionsToPackage{hyphens}{url}\usepackage{hyperref}

\newtheorem{proposition}{Proposition}
\newtheorem*{Simulation1*}{Simulation 1}
\newtheorem*{Simulation2*}{Simulation 2}
\newtheorem*{Simulation22*}{Simulation 2}
\newtheorem*{PermutationTestA*}{Permutation Test A}
\newtheorem*{PermutationTestB*}{Permutation Test B}
\newtheorem*{PermutationTestC*}{Permutation Test C}

\begin{document}

\title{Non-parametric regression for networks}

\author{Katie E.~Severn, Ian L.~Dryden and Simon P.~Preston\\
University of Nottingham, University Park, Nottingham, NG7 2RD, UK}

\date{}

%{Keywords: consistency, graph Laplacian, manifold, metric, Nadaraya-Watson}

\maketitle

\begin{abstract}
 Network data are becoming increasingly available, and so there is a need to develop suitable 
 methodology for statistical analysis. Networks can be represented as graph Laplacian matrices, which are
 a type of manifold-valued data. Our main objective is to estimate a regression curve from a sample of graph Laplacian matrices conditional on a set of Euclidean covariates, for example in dynamic networks 
 where the covariate is time. We develop an adapted 
 Nadaraya-Watson estimator which has
 uniform weak consistency for estimation using Euclidean and power Euclidean metrics. 
 We apply the methodology to the Enron email corpus to model smooth trends in monthly networks and highlight anomalous networks. Another motivating application is
 given in corpus linguistics, which explores trends in an author’s writing style over time based on word co-occurrence networks.

\end{abstract}

\section{Introduction} \label{section:IntroApplication}
Networks are of wide interest and able to represent many different phenomena,
for example social interactions and connections between regions in the brain \citep{kolaczyk2009statistical,ginestet2017hypothesis}. 
The study of dynamic networks has recently increased as more data of this type is becoming available \citep{rastelli_latouche_friel_2018}, where networks evolve over time. 
In this paper we develop some non-parametric regression methods for modelling and predicting networks where covariates are available. An
application for this work
is the study of dynamic networks derived from the Enron email corpus, in which each network corresponds to communications between employees in a particular month \citep{Diesner2005}. 
Another motivating application is the study of evolving writing styles in
the novels of Jane Austen and Charles Dickens, in which each network is a representation of a novel 
based on word co-occurences, and the covariate is the time that writing of the novel began
\citep{Severnetal19}. 
In both applications, the goal is to model smooth trends in the structure of the dynamic networks as they evolve over time. 

An example of previous work on dynamic network data is \citet{friel2016interlocking} who embedded nodes of bipartite dynamic networks in a latent space, motivated by networks representing the connection of leading Irish companies and board directors. 
We will also use embeddings in our work, although the bipartite constraints are not present. Further approaches include object functional principal components analysis \citep{Dubeymueller19} applied to time-varying networks from New York taxi trip data; multi-scale time-series modelling \citep{Kangetal17} applied to
magnetoencephalography data in neuroscience;  
and quotient space space methodology applied to brain arterial networks \citep{Guosriv20}.

The analysis of networks is a type of object oriented data analysis \citep{Marralon14}, and important considerations are to decide what are the data objects and how are they represented.  
We consider datasets where each observation is
a weighted network, denoted $G_m=(V,E)$, comprising a set of nodes, $V=\lbrace v_1, v_2,\ldots, v_m\rbrace$, and
a set of edge weights, $E=\lbrace w_{ij} : w_{ij}\geq 0, 1\leq i,j \leq m\rbrace$, indicating nodes $v_i$ and
$v_j$ are either connected by an edge of weight $w_{ij}>0$, or else unconnected (if $w_{ij}=0$). An unweighted
network is the special case with $w_{ij}\in\lbrace 0,1\rbrace$.  We restrict attention to networks that are
undirected and without loops, so that $w_{ij}=w_{ji}$ and $w_{ii}=0$, then any such network can be identified with
its graph Laplacian matrix $\mathbf{L}=(l_{ij})$, defined as
\begin{align*}
   l_{ij} = 
\begin{cases}
    -w_{ij}, & \text{if } i\neq j\\
    \sum_{k\neq i}w_{ik},& \text{if } i=j
\end{cases}
\end{align*}  
for $1\leq i,j \leq m$. 
The graph Laplacian matrix can be written as $\mathbf{L}=\mathbf{D}-\mathbf{A}$, in terms of the
adjacency matrix, $\mathbf{A}=(w_{ij})$, and degree matrix 
$\mathbf{D}=\text{diag}(\sum_{j=1}^mw_{1j},\ldots,\sum_{j=1}^mw_{mj})=\text{diag}(\mathbf{A}\mathbf{1}_m)$, where
$\mathbf{1}_m$ is the $m$-vector of ones.  The $i$th diagonal element of $\mathbf{D}$ equals the degree of node
$i$. The space of $m \times m$ graph Laplacian matrices is of dimension $m(m-1)/2$ and is
\begin{align}\label{eq:lapl space}
\mathcal{L}_m=\lbrace \mathbf{L}=(l_{ij}):\mathbf{L}=\mathbf{L}^T ;\, l_{ij}\leq 0 \, \forall i\neq j ;\, \mathbf{L} {\mathbf{1}_m}={\mathbf{0}_m} \rbrace,
\end{align}
where $\mathbf{0}_m$ is the $m$-vector of zeroes. 
In fact the space $\mathcal{L}_m$ is a closed convex subset of the cone of centred 
symmetric 
positive semi-definite $m\times m$ matrices
%\begin{align}
%\mathcal{PSD}^*_m=\lbrace \mathbf{S}^{m\times m} : x^T\mathbf{S}x\geq 0 \, \forall x %\in \mathbb{R}^m; \, \mathbf{S}=\mathbf{S}^T; \mathbf{S} %{\mathbf{1}_m}={\mathbf{0}_m} \rbrace , 
%\end{align}
and $\mathcal{L}_m$ is a manifold with corners \citep{ginestet2017hypothesis}. 

Since the sample space $\mathcal{L}_m$ for graph Laplacian data is non-Euclidean, standard approaches to non-parametric regression cannot be applied directly. In this paper, we use the statistical framework introduced in \citet{Severnetal19} for extrinsic analysis of graph Laplacian data, in which ``extrinsic'' refers to the strategy of mapping data into a Euclidean space, where analysis is performed, before mapping back to $\mathcal{L}_m$. The choice of mapping enables freedom in the choice of metric used for the statistical analysis, and in various applications with manifold-valued data analysis there is evidence of advantage in using non-Euclidean metrics \citep{Drydenetal09,Pigolietal14}.

A summary of the key steps in the extrinsic approach is: 
\begin{enumerate}[i)]
\item embedding to another manifold ${\mathcal M}_m$ by raising the graph Laplacian matrix to a power,
\item mapping from $\mathcal{M}_m$ to a (Euclidean) tangent space $T_\nu(\mathcal{M}_m)$ in which to carry out statistical analysis, 
\item inverse mapping from the tangent space $T_\nu(\mathcal{M}_m)$ to the embedding manifold $\mathcal{M}_m$,
\item reversing the powering in i), and projecting back to graph Laplacian space $\mathcal{L}_m$, 
\end{enumerate} 
which we explain in more detail as follows. 
First write $\mathbf{L}=\mathbf{U}\boldsymbol{\Xi}\mathbf{U}^{T}$ by the
spectral decomposition theorem, with $\boldsymbol{\Xi}=\text{diag}(\xi_1,\ldots,\xi_m)$ and
$\mathbf{U}=(\mathbf{u}_1,\ldots,\mathbf{u}_m)$, where $\lbrace\xi_i\rbrace_{i=1,\ldots,m}$  are the eigenvalues, which are non-negative as any $\mathbf{L}$ is positive semi-definite, and 
$\lbrace\mathbf{u}_i\rbrace_{i=1,\ldots,m}$ are the corresponding 
eigenvectors of $\mathbf{L}$. We consider the following map which raises the graph Laplacian to the power $\alpha>0$:
\begin{align}
\begin{split}
\text{F}_\alpha(\mathbf{L})&= \mathbf{L}^\alpha=\mathbf{U}\boldsymbol{\Xi}^\alpha\mathbf{U}^T : \mathcal{L}_m\rightarrow \operatorname{Image}(\mathcal{L}_m) \subset \mathcal{M}_m.\\
\end{split}  \label{Fmap}
\end{align}
In this paper we take ${\cal M}_m$ to be the Euclidean space of symmetric $m \times m$ matrices. In terms of $\text{F}_\alpha(\cdot)$ we define the power Euclidean distance between two graph Laplacians as
\begin{equation}
d_\alpha( \mathbf{L}_1 , \mathbf{L}_2 ) =   \| \text{F}_\alpha(\mathbf{L}_1) - \text{F}_\alpha(\mathbf{L}_2) \|,
\label{eqn:power:euclidean:distance}
\end{equation}
where $\| \mathbf{X} \| = \sqrt{ {\rm trace}(\mathbf{X}^T \mathbf{X})  }$  is the Frobenius norm, also known as the Euclidean norm. For the special case $\alpha=1$, (\ref{eqn:power:euclidean:distance}) is just the Euclidean distance. \citet{Severnetal19} further considered a Procrustes distance 
which includes minimization over an orthogonal matrix,  
approximately allowing relabelling of the nodes, and in this case the 
embedding manifold $\mathcal{M}_m$ is a Riemannian manifold known as the 
size-and-shape space \citep[p.99]{Drydmard16}. However, 
in this paper for simplicity and because the labelling is known, we shall just consider the power Euclidean metric. 
 
After the embedding the data are mapped to a tangent space $T_\nu(\mathcal{M}_{m})$ of the embedding manifold at $\nu$ using the bijective transformation 
$$\pi_\nu^{-1} :  \mathcal{M}_{m} \to T_\nu(\mathcal{M}_{m}) . $$
In our applications we take $\nu=0$ and the tangent co-ordinates are
$\mathbf{v} = {\rm vech}(\mathbf{H} \text{F}_\alpha(\mathbf{L}) \mathbf{H}^T)$, 
where ${\rm vech}()$ is the vector of elements above and including the 
diagonal and $\mathbf{H}$ is the Helmert sub-matrix 
\citep[p49]{Drydmard16}. The main point of note is that the tangent space is a Euclidean space of dimension $m(m-1)/2$.  Hence, multivariate linear statistical analysis can be carried out in $T_\nu(\mathcal{M}_{m})$; for example,  \citet{Severnetal19} considered estimation, 
two-sample hypothesis tests, and linear regression. 

After carrying out statistical analysis, the fitted values are then transformed back to the graph Laplacian 
space as follows:
\begin{equation}
P_{\cal L} \circ \text{G}_\alpha \circ \pi_\nu :  T_\nu(\mathcal{M}_{m}) \to \mathcal{L}_m  \label{uniquemap}
\end{equation}
where $\text{G}_\alpha:  \mathcal{M}_{m}  \to  \mathcal{M}_{m}$ is a map that reverses the power, and $P_{\cal L} :  \mathcal{M}_{m}  \to  \mathcal{L}_{m} $ is the projection to the closest point in graph Laplacian space using Euclidean 
distance. The projection $P_{\cal L}$ is obtained by solving a convex optimization problem using quadratic programming, and the solution is therefore unique. 
See \citet{Severnetal19} for full details of this framework. 

For visualising results, it is useful to map the data and fitted regression lines in $\mathcal{L}_m$
into $\mathbb{R}^2$. In this paper we do so using principal component analysis (PCA) such that the two plotted dimensions reflect the two orthogonal dimensions of greatest sample variability in the tangent space \citep{Severnetal19}. 

%[is this correct and adequate, or do we need to say more?]

%It is often useful the summarize the variability in a sample of networks in a lower dimensional space, and we consider
%principal components analysis (PCA) in the tangent space which is the same as multidimensional scaling for data with 
%Euclidean distances \citep{bibby1979multivariate}. 
%In this paper we provide plots of the first two principal components scores to give a low dimensional summary of the network data. 

%We also provide an adaptation of the plots for widely dispersed dynamic network data to reduce the horseshoe effect that commonly arises. 

\section{Nadaraya-Watson estimator for network data}
\subsection{Nadaraya-Watson estimator}
We first review the classical 
Nadaraya-Watson estimator  \citep{10.2307/25049340,doi:10.1137/1110024} before defining an analogous version for data on $\mathcal{L}_m$. Consider the regression problem where 
we want to predict an unknown variable $y(\mathbf{x}) \in \mathbb{R}$ with known covariate $\mathbf{x} \in \mathbb{R}^p$ for the dataset 
of independent and identically distributed random variables  
$(\lbrace Y_1, \mathbf{X}_1\rbrace,\dots, \lbrace Y_n,\mathbf{X}_n\rbrace)$ observed at   $(\lbrace y_1, \mathbf{x}_1\rbrace,\dots, \lbrace y_n,\mathbf{x}_n\rbrace)$
with $E\{ |Y|\} < \infty$. 
The aim is to estimate the regression function 
$$ m(\mathbf{x}) = E [ y(\mathbf{x}) | \mathbf{x} ] . $$
The Nadaraya-Watson estimator is 
\begin{align}\label{eq: NW estimator}
\hat m(\mathbf{x})=\frac{\sum_{i=1}^nK_h(\mathbf{x}-\mathbf{x}_i)y_i}{\sum_{i=1}^nK_h(\mathbf{x}-\mathbf{x}_i)},
\end{align}
where $K_h \ge 0$ is a kernel function with bandwidth $h>0$. 

Consider now a version of the regression problem with $y(\mathbf{x})$ replaced with an $m \times m$  graph Laplacian matrix $\mathbf{L}(\mathbf{x}) \in \mathcal{L}_m$ 
with known covariates $\mathbf{x} \in \mathbb{R}^p$ and dataset $(\lbrace \mathbf{L}_1, \mathbf{x}_1\rbrace,\dots, \lbrace \mathbf{L}_n,\mathbf{x}_n\rbrace)$ with 
$E\{ |(\mathbf{L})_{ij}|\} < \infty, i=1,\ldots,m; j=1,\ldots,m$. 
We wish to 
estimate the regression function 
$$ \boldsymbol{\Lambda}(x) = E [ \mathbf{L}(\mathbf{x}) | \mathbf{x} ] . $$
A natural analogue of $\hat{m}(\mathbf{x})$ in \eqref{eq: NW estimator} for graph Laplacian data given covariate, $\mathbf{x} \in \mathbb{R}^p$, is 
\begin{align}\label{eq:euclidean nadaraya}
\hat{\mathbf{L}}_{NW}(\mathbf{x})=\frac{\sum_{i=1}^nK_h(\mathbf{x}-\mathbf{x}_{i})\mathbf{L}_i}{\sum_{i=1}^nK_h(\mathbf{x}-\mathbf{x}_{i})} 
= \sum_{i=1}^n W_{hi}(\mathbf{x})  \mathbf{L}_i , 
\end{align}
where 
$$W_{hi}(\mathbf{x}) = \frac{ K_h(\mathbf{x}-\mathbf{x}_{i}) }  {\sum_{i=1}^nK_h(\mathbf{x}-\mathbf{x}_{i})} 
\ge 0$$ 
and note that $\sum_{i=1}^n W_{hi}(\mathbf{x}) = 1$. 
A common choice of kernel function is the Gaussian kernel 
 \begin{align}\label{eq:kernel function}
 K_h(\mathbf{u})= \frac{1}{h\sqrt{2\pi}}\exp\left(-\frac{\Vert \mathbf{u}\Vert^2 }{2h^2}\right),
 \end{align}
 which is bounded above and strictly positive for all $\mathbf{u}$.  We use a truncated 
 version of \eqref{eq:kernel function} 
such that $K_h(\mathbf{u}) = 0$ for $\Vert \mathbf{u} \Vert > c $ (with $c$ large) in order that this truncated kernel has compact support, as required by theoretical results presented later. 
Wherever $\hat{\mathbf{L}}_{NW}$ is defined (meaning that at least one of the 
$K_h(\mathbf{x}-\mathbf{x}_{i})$ is non-zero) it is a sum of positively weighted graph Laplacians. Since 
the space $\mathcal{L}_m$ is a convex cone \citep{ginestet2017hypothesis}, itself defined as the sum of positively weighted graph Laplacians, thus  $\hat{\mathbf{L}}_{NW}(\mathbf{t})\in \mathcal{L}_m$ as required.

The estimator in (\ref{eq:euclidean nadaraya}) can equivalently be written 
as the graph Laplacian that minimises a weighted sum of squared Euclidean, $d_1$, distances to the sample data
\begin{align}\label{eq:manifold nadaraya euc}
\hat{\mathbf{L}}_{NW}(\mathbf{x})
=\arg\inf_{{\mathbf{L}}\in \mathcal{L}_m}\sum_{i=1}^n  W_{hi}(\mathbf{x}) d_1(\mathbf{L}_i, \boldsymbol{\mathbf{L}})^2 = 
\arg\inf_{{\mathbf{L}}\in \mathcal{L}_m}\sum_{i=1}^n  W_{hi}(\mathbf{x})
\| \mathbf{L}_i - \boldsymbol{\mathbf{L}} \|^2 ,
\end{align}
using weighted least squares. In principle, the Euclidean distance $d_1$ in 
$\eqref{eq:manifold nadaraya euc}$ can be replaced with a different distance metric, $d$, though solving for the estimator entails an 
optimisation on the manifold $\mathcal{L}_m$, which can be 
theoretically and computationally challenging. Hence instead we generalise to other distances via an extrinsic approach, and define
%We can generalise this to give a more general Nadaraya-Watson estimate %suitable for minimising any distance between graph Laplacians %\citep{Davis2010}. This general Nadaraya-Watson estimate is the projected %matrix that minimises the given distance, $d$, between weighted graph %Laplacians, given as,
\begin{align}\label{eq:manifold nadaraya}
\hat{\mathbf{L}}_{NW,d}(\mathbf{x})
=\text{P}_\mathcal{L}
\left(\arg\inf_{\boldsymbol{\mathbf{L}}\in \mathcal{M}_m}\sum_{i=1}^n W_{hi}(\mathbf{x})  d(\mathbf{L}_i, \boldsymbol{\mathbf{L}})^2\right),
\end{align}
which is simpler provided, as here, the embedding manifold $\mathcal{M}_m$ is chosen such that the optimisation is straightforward. 
The projection is needed to map back to graph Laplacian space ${\cal L}_m$. 
For the power Euclidean metric, $d_\alpha$, consider the Nadaraya-Watson estimator in the tangent space $T_\nu({\cal M}_m)$,
\begin{equation}
\hat{\mathbf{L}}_{ {\cal M}_m ,\alpha} (\mathbf{x})
= \sum_{i=1}^n W_{hi}(\mathbf{x})  \pi_\nu^{-1}( \text{F}_\alpha ( \mathbf{L}_i) )  , 
\end{equation}
in terms of which, after mapping back to ${\mathcal L}_m$ using (\ref{uniquemap}), 
the resulting Nadaraya-Watson estimator in the graph Laplacian space is
\begin{align}
\hat{\mathbf{L}}_{NW,\alpha}(\mathbf{x})
=   P_{\cal L} \circ \text{G}_\alpha \circ \pi_\nu  \left(  \hat{\mathbf{L}}_{{\cal M}_m,\alpha}(\mathbf{x})  \right).   \label{powerNW}
\end{align}
When $\alpha=1$ this simplifies to (\ref{eq:euclidean nadaraya}).

\subsection{Uniform weak consistency}

First we show that the Nadaraya-Watson estimator for graph Laplacians (\ref{eq:euclidean nadaraya}) is uniformly weakly consistent. 

Let 
\begin{equation}
J_{n} = E \left\{ \int \left\| \hat {\mathbf{L}}_{NW}(\mathbf{x}) - \boldsymbol\Lambda(\mathbf{x}) \right\|^2 \mu(\mathrm{d}\mathbf{x}) \right\} 
\end{equation}
where $\mu$ is the probability measure of $\mathbf{x}$.  

\begin{proposition} Suppose the kernel function $K_h$ is 
  non-negative on $\mathbb{R}^p$, bounded above, has 
  compact support and is strictly positive in a neighbourhood of the origin. 
  %$K_h \ge \beta
  %I_B$ for some $\beta>0$ and some closed sphere $B$ %centred at the origin and
  %having positive radius, with $I$ the indicator %function.  
  If $E\{ \|
  \mathbf{L} \|^2 \} < \infty$, as $h \to 0$ and $n h^d \to \infty$ it follows
  that $J_{n} \to 0$. Hence the Nadaraya-Watson estimator  $\hat
  {\mathbf{L}}_{NW}(\mathbf{x})$ is uniformly weakly consistent for the true
  regression function $\boldsymbol\Lambda(\mathbf{x})$.  \end{proposition}

\begin{proof} {\rm 
Consider the univariate regression problem for the $(i,j)$th element of $\hat{\mathbf{L}}_{NW}(\mathbf{x})$. From \citet{devroye1980} we know
that under the conditions of the proposition we have 
$$J_{n,ij} = E \left\{ 
\int \left| \left( \hat{\mathbf{L}}_{NW}(\mathbf{x})\right)_{ij} - \left(\boldsymbol\Lambda(\mathbf{x})\right)_{ij} \right|^2 \mu(\mathrm{d}\mathbf{x}) \right\} \to 0 \; , \;  i=1,\ldots,m; j=1,\ldots,m,$$
as $h \to 0$ and $n h^d \to \infty$; and since   
$$ J_{n} = \sum_{i=1}^m\sum_{j=1}^m J_{n,ij},$$
thus $J_{n} \rightarrow 0$. 
} \end{proof}

The result can be extended to the power Euclidean distance based Nadaraya-Watson estimator (\ref{powerNW}). 
Let 
\begin{equation}
J_{n,\alpha} = E \left\{ \int \left\|  \hat{\mathbf{L}}_{NW,\alpha}(\mathbf{x}) - \boldsymbol\Lambda(\mathbf{x}) \right\| ^2 \mu(\mathrm{d}\mathbf{x}) \right\}.
\end{equation}
where $\mu$ is the probability measure of $X$.  

\begin{proposition}
Under the conditions of Proposition 1 it follows that 
$J_{n,\alpha} \to 0$. Hence the power Euclidean Nadaraya-Watson estimator $\hat{\mathbf{L}}_\alpha(\mathbf{x})$ is uniformly weakly
 consistent for the true regression 
function $\boldsymbol\Lambda(\mathbf{x})$. 
\end{proposition}

\begin{proof} {\rm First embed the graph Laplacians in the Euclidean manifold
  ${\cal M}_m$ and map to a tangent space $T_\nu({\mathcal M}_m)$.  Consider
  the univariate regression problem for the $(i,j)$th element of
  $\pi_\nu^{-1}(\mathbf{F}_\alpha( \mathbf{L}(\mathbf{x}) ))$. Again from
  \citet{devroye1980,spiegelman1980} we know that under the conditions of the
  proposition we have uniform weak consistency in the tangent space.  $$
  J_{n,\alpha,ij} = E \left\{ \int \left| \left( \hat{ \mathbf{L}} _{ {\mathcal M}_m,\alpha}
  ( \mathbf{x}) \right)_{ij} -
  \pi_\nu^{-1}\left(\mathbf{F}_\alpha(\boldsymbol\Lambda(\mathbf{x}))\right)_{ij} \right|^2
  \mu(\mathrm{d}\mathbf{x}) \right\}  \to 0 \; , \;  i=1,\ldots,m; j=1,\ldots,m,$$

as $h \to 0$ and $n h^d \to \infty$. Also, using the continuous mapping theorem
and Pythagorean arguments as in \citet{Severnetal19}, we have $$ J_{n,\alpha}
\le  \sum_{i=1}^m\sum_{j=1}^m J_{n,\alpha,ij} \to 0 , $$ as $h \to 0$ and $n
h^d \to \infty$.  } \end{proof}

\subsection{Bandwidth selection}
The result that the power Euclidean Nadaraya-Watson estimator is uniformly
weakly consistent gives reassurance that the method is a sensible practical
approach to non-parametric regression for predicting networks. The result is asymptotic, however, which leaves open the question of how to choose the bandwidth, $h$, in practice.
One way to do so is to select it via cross validation \citep{Efron93} as follows. Denote by $\hat{\mathbf{L}}_{-i}(\mathbf{x}; h)$ a Nadaraya-Watson estimator at $\mathbf{x}$, based on distance metric $d$, with bandwidth $h$, trained on all the training observations excluding the $i$th. Selection of bandwidth by cross validation then involves choosing $h$ to minimise the criterion
\begin{equation}
% \hat{h} = \arg\inf_h 
\sum_{i=1}^n d\left(\mathbf{L}_i, \hat{\mathbf{L}}_{-i}(\mathbf{x}_i; h) \right)^2.
\label{eqn:CV:criterion}
\end{equation}

\section{Application: Enron email corpus}\label{sec:enron}
The Enron dataset was made public during the legal investigation of Enron by the Federal Energy Regulatory Commission \citep{klimt2004introducing} 
and an overview can be found in \citet{Diesner2005}.
Similar to \citet{shetty2004enron} we use this data to form social networks between $m=151$ employees and 
the data are available from \citet{Enrondata}. For each month we create a  
network with employees as nodes. The edges between nodes have weights that are the number of emails exchanged between the two employees in the given month.  
The networks we consider are for the whole months from June 1999 (month 1) to May 2002 (month 36), and 
we standardise by dividing by the trace of the graph Laplacian for each month.
The aim is to model smooth trends in the structure of the dynamic networks as they evolve over time, and we also wish to highlight anomalous months where the network is rather unusual compared to the fitted trend.

%The Enron dataset has been studied extensively as social networks due to the uniqueness of the dataset. Previous work on the dataset includes studying the hierarchy, clustering and importance of the employers within Enron \citep{Agarwal:2012:CGS:2390665.2390706, 4983357}. A lot of work has also explored the time structure of the Enron networks and how the email interactions change with time \citep{Diesner2005ExplorationOC}. We will focus on the time structure of the Enron networks, especially trying to identify changes within the network that correspond to the scandal.

\begin{figure}[htbp]
    \centering
     
      \begin{subfigure}[b]{0.47\textwidth}    
 \caption{}
   %  \hspace{0.5cm}
        \includegraphics[width=8.0cm]{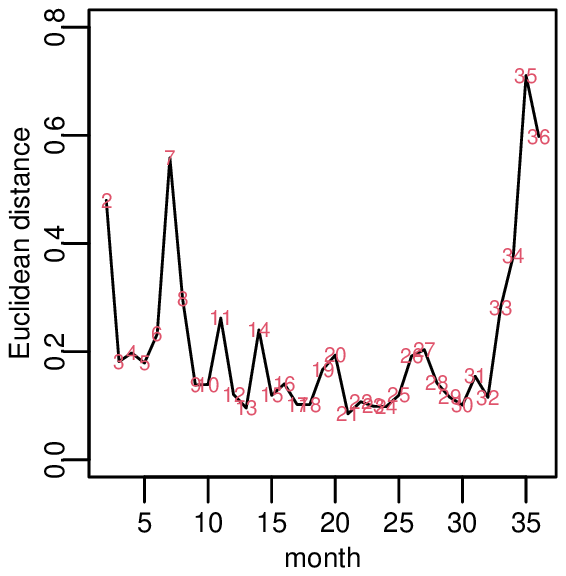}
    \end{subfigure}
        \begin{subfigure}[b]{0.47\textwidth}
 \caption{}
  %    \hspace{0.5cm}
        \includegraphics[width=8.0cm]{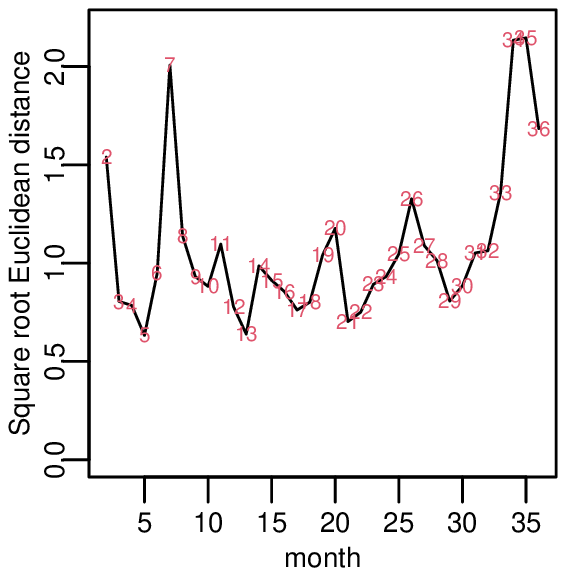}
    \end{subfigure}
            \caption{  {\textit{  Distances, $d(\mathbf{L}_{i-1}, \mathbf{L}_{i}), i=2,\ldots,36$, between consecutive observations for the monthly Enron networks for (a) the Euclidean metric, $d_1$, and (b) the square root Euclidean metric, $d_{\frac{1}{2}}$.          
          }}}\label{fig:dist}      
            \end{figure}

In Figure \ref{fig:dist} we plot the distances between consecutive monthly graph Laplacians 
using Euclidean distance (a) and square root Euclidean distance (b). 
Some of the largest successive distances are at times $1-2, 6-7, 33-34, 34-35, 35-36$, and these are possible candidate positions for anomalous networks that are 
rather different.

\begin{figure}[htbp]
    \centering
    \begin{subfigure}[b]{0.47\textwidth}
 \caption{}
        \includegraphics[width=8.0cm]{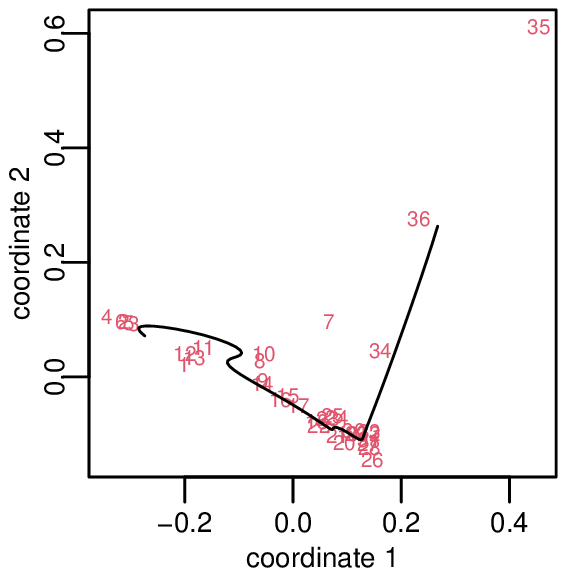}
    \end{subfigure}
        \begin{subfigure}[b]{0.47\textwidth}
         \caption{}
        \includegraphics[width=8.0cm]{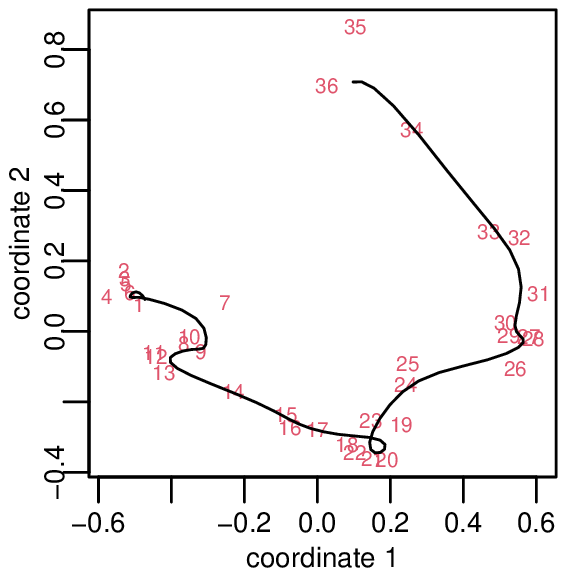}
    \end{subfigure}
        \begin{subfigure}[b]{0.47\textwidth}
         \caption{}
     \includegraphics[width=8.0cm]{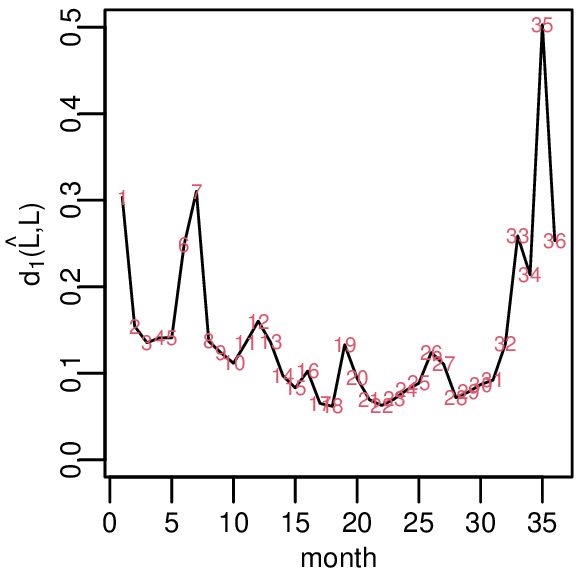}
    \end{subfigure}
            \begin{subfigure}[b]{0.47\textwidth}
         \caption{}
     \includegraphics[width=8.0cm]{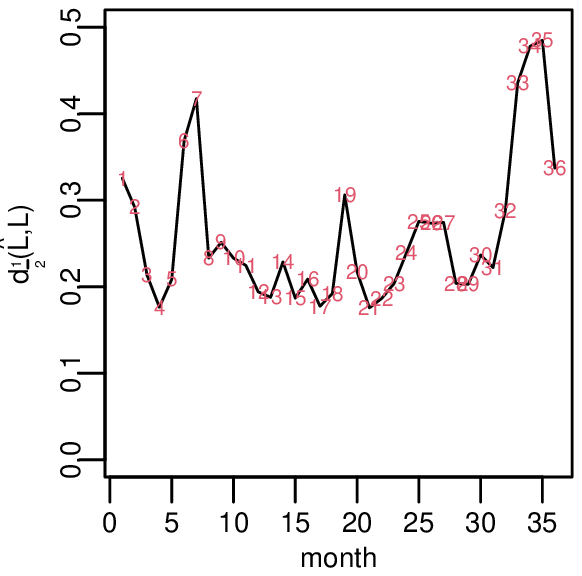}
    \end{subfigure}
    
 %    \begin{subfigure}[b]{0.3\textwidth}
 %\caption{}
 %       \includegraphics[width=5cm]{Dmat_edit_simplex_enron}
 %   \end{subfigure}
 %       \begin{subfigure}[b]{0.3\textwidth}
 %\caption{}
 %       \includegraphics[width=5cm]{Dmat_edit_power_enron}
 %   \end{subfigure}
 %           \begin{subfigure}[b]{0.3\textwidth}
 %                \caption{}
 %       \includegraphics[width=5cm]{simplex_rho_max_pc1}
 %   \end{subfigure}
 %       \begin{subfigure}[b]{0.3\textwidth}
 %\caption{}
 %       \includegraphics[width=5cm]{power_rho_max_pc1}
 %   \end{subfigure}
          \caption{  {\textit{ PCA plots showing the data and Nadaraya-Watson curves, and residual plots for the Enron network data. In each plot the red digits indicate the observation number (month index). In the upper plots the black lines show the Nadaraya-Watson regression curves in the space of the first two principal components, using (a) the Euclidean metric, $d_1$, and $h=2$; (b) the Square root Euclidean metric, $d_{\frac{1}{2}}$, and $h=1$.  We performed the calculations for a various values $h=0.5,1,2,4,8$, and the chosen value of $h$ was whichever that was optimal with respect to \eqref{eqn:CV:criterion}. Plots (c) and (d) are corresponding residual plots showing distance $d(\hat{\mathbf{L}}_i,{\mathbf{L}}_i) $ between the fitted values $\hat{\mathbf{L}}_i$ and the observations ${\mathbf{L}}_i$.
          %The red digits indicate the month of the data.   MDS plots using the Mahalanobis metric for (c) the Enron data with $\alpha=1$ and (d) the Enron data with $\alpha=\frac{1}{2}$ 
         % using an overall estimate from $\rho$, 
          %and the Mahalanobis metric when $\rho$ is estimated to maximize the variance explained by PC1 we have (e) the Enron data with $\alpha=1$ and (f) the Enron data with $\alpha=\frac{1}{2}$.   In (e) and (f) the two colours indicate two clusters.      
          }}}\label{fig:pca enron}      
            \end{figure}

\begin{figure}[htbp]
    \centering
     \begin{subfigure}[b]{0.47\textwidth}
 \caption{}
        \includegraphics[width=8.0cm]{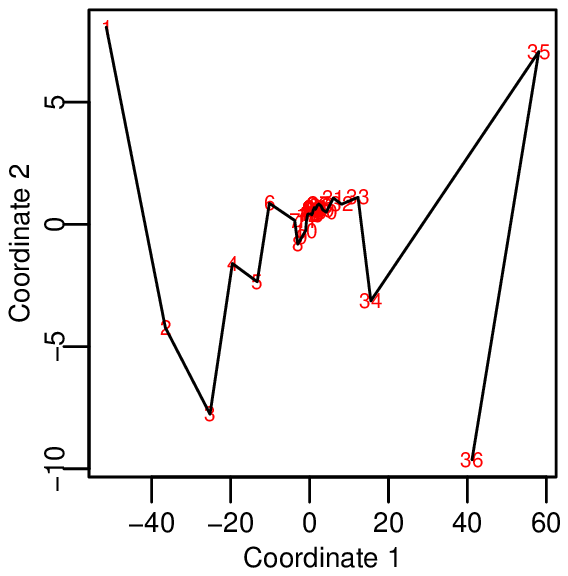}
    \end{subfigure}
        \begin{subfigure}[b]{0.47\textwidth}
 \caption{}
        \includegraphics[width=8.0cm]{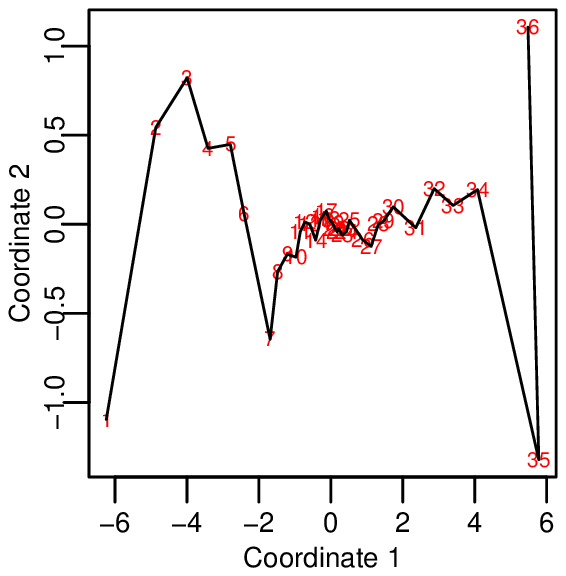}
    \end{subfigure}
            \begin{subfigure}[b]{0.47\textwidth}
                 \caption{}
        \includegraphics[width=8.0cm]{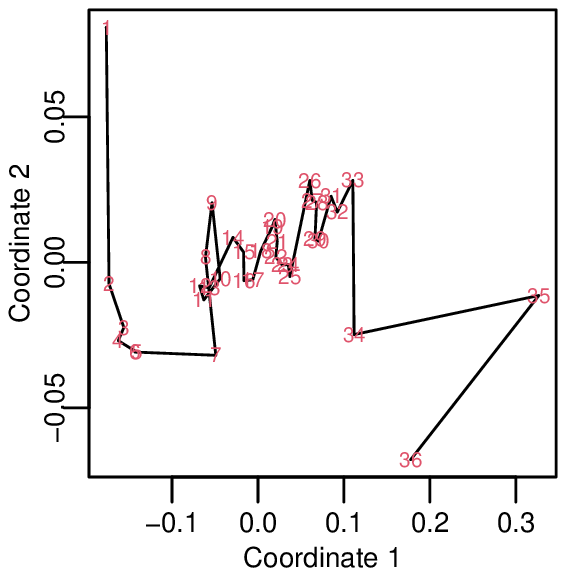}
    \end{subfigure}
        \begin{subfigure}[b]{0.47\textwidth}
 \caption{}
        \includegraphics[width=8.0cm]{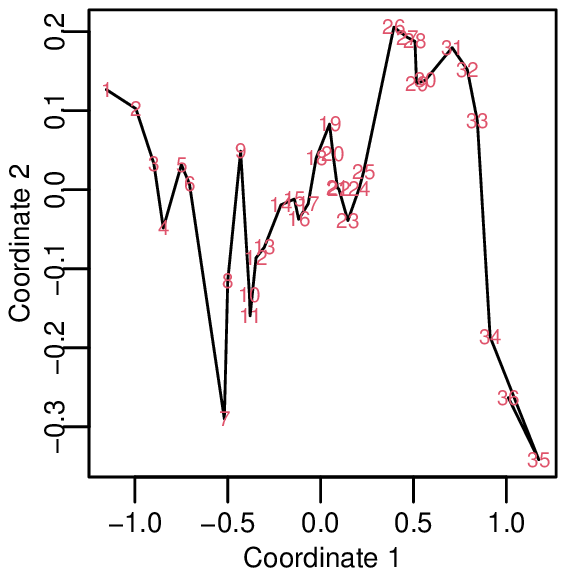}
    \end{subfigure}
          \caption{  {\textit{ 
          The red digits indicate the month of the data.   MDS plots using the Mahalanobis metric for (a) the Enron data with $\alpha=1$ and (b) with $\alpha=\frac{1}{2}$ 
          using an overall estimate from $\rho$, 
          and the Mahalanobis metric when $\rho$ is estimated to maximize the variance explained by PC1 for (c) the Enron data with $\alpha=1$ and (d)  with $\alpha=\frac{1}{2}$.     
          }}}\label{fig:edit mds enron}      
            \end{figure}

We provide a PCA plot of the first two PC scores in Figure \ref{fig:pca enron}(a),(b) and include the Nadaraya-Watson estimator projected into 
the space of the first two PCs. Here the bandwidth has been chosen by cross-validation as $h=2$ for the Euclidean case and $h=1$ for the square root metric. The  Nadaraya-Watson estimator provides a smooth path through the data, and the structure is clearer in the square root metric plot.

We are interested in finding anomalies in the Enron dynamic networks and so 
we compute the distances from each network to the fitted value from the 
Nadaraya-Watson estimate.
Figure \ref{fig:pca enron} shows these residual distances of 
each graph Laplacian to the fitted Nadaraya-Watson values for (c) the
Euclidean metric and (d) the square root metric. 
Some of the largest residuals are months 1,7,35 for Euclidean and 7,33,34,35 for the square root metric, and these are candidates for anomalies.

From Figure \ref{fig:pca enron}(b) it looks like there is an approximate horseshoe shape in the PC score plot which could be 
an example of the horseshoe effect  \citep{kendall1971abundance, diaconis2008horseshoes, morton2017uncovering}. 
 We might conclude there is a change point in the data around months 20-26 from these plots but this may be misleading \citep{doi:10.1098/rsta.1970.0091}. 
Explained in \citet[page 412]{Mardiaetal79}, the horseshoe effect occurs when the distances which are ``large'', between data points, appear the same as those that are ``moderate''. \citet{morton2017uncovering} described this as a ``saturation property'' of the metric, and so on the PCA plot the point corresponding to a `large' time is pulled in closer to time 1 than we intuitively would expect.

As an alternative to PCA, which seeks to address this horseshoe effect, we consider multidimensional scaling (MDS) with a Mahalanobis metric in the tangent space \citep[p.31]{Mardiaetal79} between two graph Laplacians $\mathbf{L}_k$ and $\mathbf{L}_l$, at times $k$ and $l$ respectively, which is:
\begin{align*}
\sqrt{(\pi_0^{-1}(\text{F}_\alpha(\mathbf{L}_k))-\pi_0^{-1}(\text{F}_\alpha(\mathbf{L}_l))-\boldsymbol{\mu})^T\boldsymbol{\Sigma}_{kl}^{-1}(\pi_0^{-1}(\text{F}_\alpha(\mathbf{L}_k))-\pi_0^{-1}(\text{F}_\alpha(\mathbf{L}_l))-\boldsymbol{\mu})},
\end{align*}
where $\boldsymbol{\mu}$ and $\boldsymbol{\Sigma}_{kl}$ are the mean and covariance matrix of $\pi_0^{-1}(\text{F}_\alpha(\mathbf{L}_k))-\pi_0^{-1}(\text{F}_\alpha(\mathbf{L}_l))$ respectively. Here we take $\boldsymbol{\mu}$ as zero and  consider an isotropic AR(1) model which has covariance matrix 
$$
\boldsymbol{\Sigma}_{kl}=\frac{\sigma^2\rho^{\vert k-l \vert}}{1-\rho}\mathbf{I}_\frac{m(m-1)}{2},
$$
which is a diagonal matrix where the diagonal elements are the variance of elements and we have assumed a 0 covariance between any other elements.
Writing $\mathbf{y}_k=\pi_0^{-1}(\text{F}_\alpha(\mathbf{L}_k))$ and $\mathbf{v}_k=\pi_0^{-1}(\text{F}_\alpha(\mathbf{L}_{k-1}))$ we estimate 
$\rho$ by least squares
$
\rho
=\frac{\sum_{k=2}^{n}(\mathbf{y}_k^T\mathbf{v}_k)}{\sum_{k=2}^{n}(\mathbf{v}_k^T\mathbf{v}_k)},
$
and we take $\sigma = 1$ as this is just an overall scale parameter. 
The Mahalanobis metric between graph Laplacians, $\mathbf{L}_k$ and $\mathbf{L}_l$, can now be written as
\begin{align*}
&=\sqrt{\frac{1-\rho}{\rho^{\vert k-l \vert}}(\pi_0^{-1}((\text{F}_\alpha(\mathbf{L}_k))-\pi_0^{-1}(\text{F}_\alpha(\mathbf{L}_l)))^T(\pi_0^{-1}(\text{F}_\alpha(\mathbf{L}_k))-\pi_0^{-1}(\text{F}_\alpha(\mathbf{L}_l)))}\\
&=\sqrt{\frac{1-\rho}{\rho^{\vert k-l \vert}}} \Vert \pi_0^{-1}(\text{F}_\alpha(\mathbf{L}_k))-\pi_0^{-1}(\text{F}_\alpha(\mathbf{L}_l)) \Vert 
=\sqrt{\frac{1-\rho}{\rho^{\vert k-l \vert}}} \Vert \text{F}_\alpha(\mathbf{L}_k)- \text{F}_\alpha(\mathbf{L}_l) \Vert 
 = \sqrt{\frac{1-\rho}{\rho^{\vert k-l \vert}}} d_{\alpha}( \mathbf{L}_k , \mathbf{L}_l )  .
\end{align*}

The plot of MDS with the Mahalanobis distance are given in  Figure \ref{fig:edit mds enron}(a)-(b). In both plots there are large distances between the 
first few and last few observations compared to the central observations, which is broadly in keeping with Figure \ref{fig:dist}(a),(b) although 
the middle observations do seem too close together in the MDS plots. We consider
an alternative estimate in choosing $\rho$ that maximises the variance explained by the first PC scores for each example, shown in Figure \ref{fig:edit mds enron}(c)-(d).
These final MDS plots are more in agreement with the distance plots of  Figure \ref{fig:dist}. In particular in Figure \ref{fig:edit mds enron}(d) we see that months 7, 34, 35, 36 look rather different from the rest.

%For the Enron data using the Euclidean metric $\rho=0.95$ for the square root metric $\rho=0.92$.
Finally we consider the main features of all the results from Figure \ref{fig:dist}-Figure \ref{fig:edit mds enron} and we see that 
the 7th, 34th and 35th months stand out as strong anomalies. The 7th month corresponds to December 1999, and this is picked out to be an anomaly in \citet{wang2014locality}, believed to coincide with Enron's tentative sham energy deal with Merrill Lynch created to meet profit expectations and boost the stock price. Month 34 and 35 correspond to March and April 2002 these correspond to the former Enron auditor, Arthur Andersen, being indicted for obstruction of justice \citep{guardianEnron}.

      \section{Application: 19th century novel networks}  \label{ex: Nw novels}
      We consider an application where it is of
      interest to analyze dynamic networks from 
      the novels of Jane Austen and Charles Dickens. 
      The 7 novels of Austen and 16 novels of Dickens were represented as samples of network data by \citet{Severnetal19}. Each novel is represented by a network where 
each node of the network is a word, and edges are formed with weights proportional to the number of times a pair of words co-occurs closely in the text. 
For each novel we produce a network counting pairwise word co-occurrences, and words are said to co-occur if they appear within five words of
each other in the text. A choice that needs to be made is if we allow co-occurrences over sentence boundaries and chapter boundaries \citep[Section 3]{evert2008corpora}, and for this dataset we allow it. 
The data are obtained from
CLiC \citep{doi:10.3366/cor.2016.0102}.

 \begin{figure}[htbp]
    \centering
     \begin{subfigure}[b]{0.49\textwidth}
 \caption{ }
        \includegraphics[trim={0 1cm 0 0.5},clip]{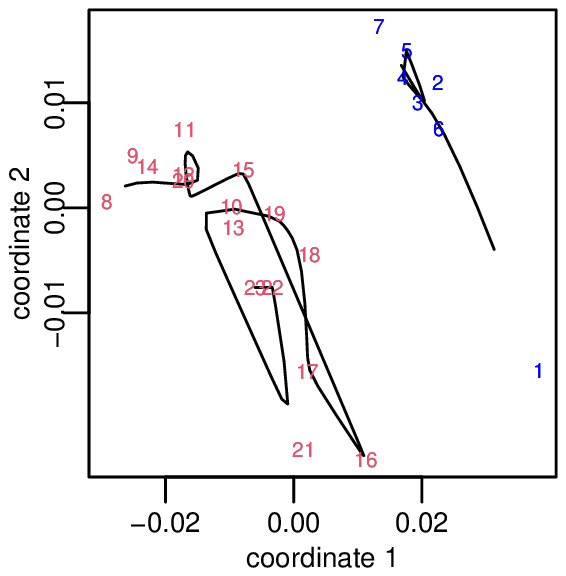}
    \end{subfigure}
          \begin{subfigure}[b]{0.45\textwidth}
 \caption{ }
        \includegraphics[trim={0 1cm 0 0.5},clip]{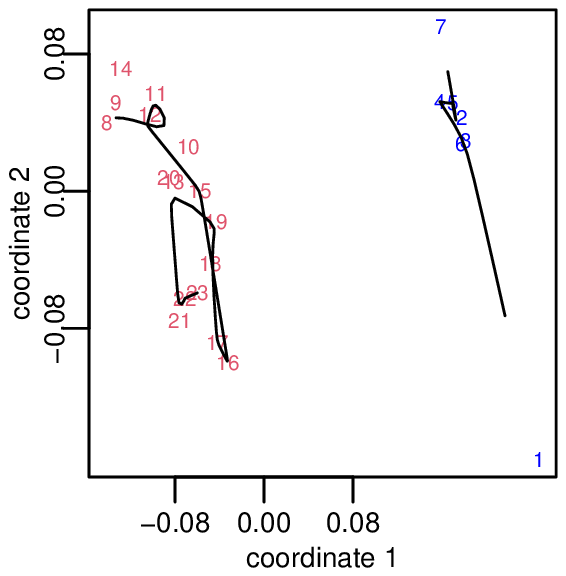}
    \end{subfigure}
  \begin{subfigure}[b]{0.49\textwidth}
 \caption{ }
        \includegraphics[trim={0 1cm 0 0.5},clip]{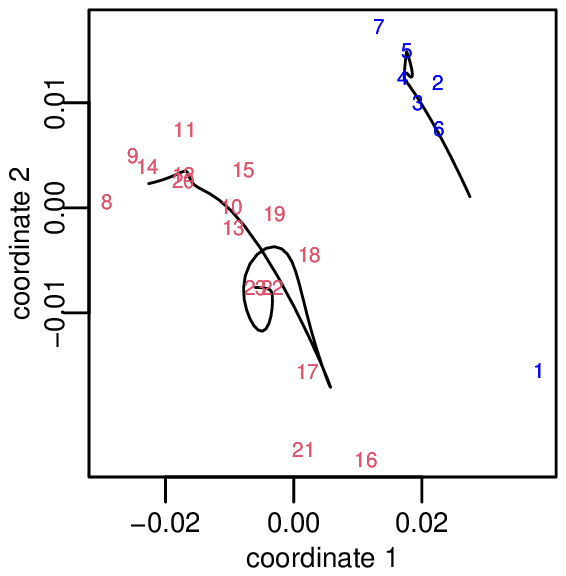}
    \end{subfigure}
          \begin{subfigure}[b]{0.45\textwidth}
 \caption{ }
        \includegraphics[trim={0 1cm 0 0.5},clip]{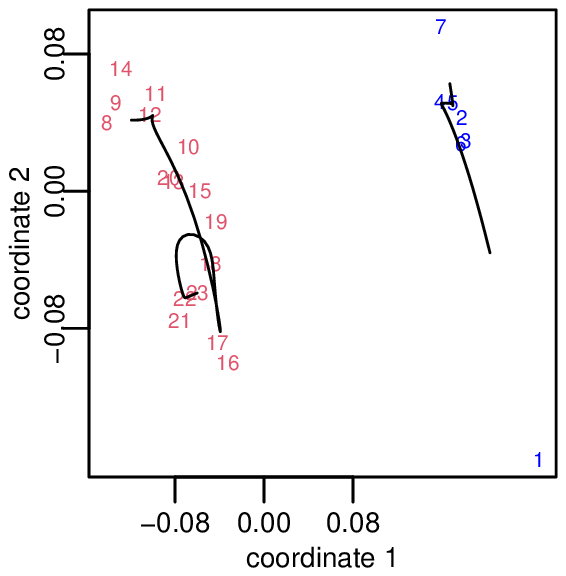}
    \end{subfigure}
      \begin{subfigure}[b]{0.49\textwidth}
 \caption{ }
        \includegraphics[trim={0 1cm 0 0.5},clip]{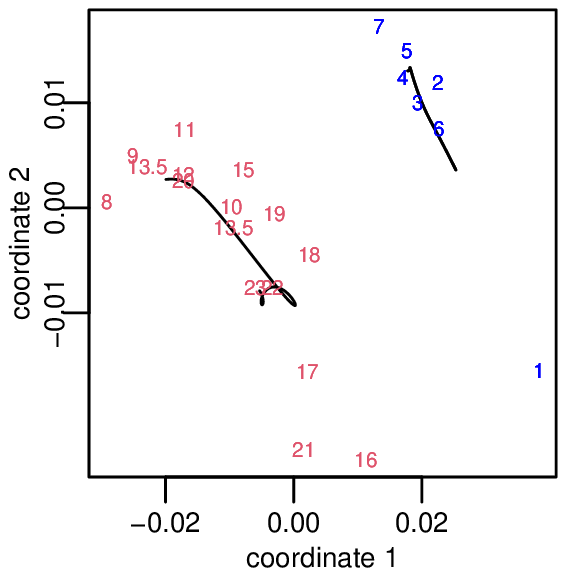}
    \end{subfigure}
          \begin{subfigure}[b]{0.45\textwidth}
 \caption{ }
        \includegraphics[trim={0 1cm 0 0.5},clip]{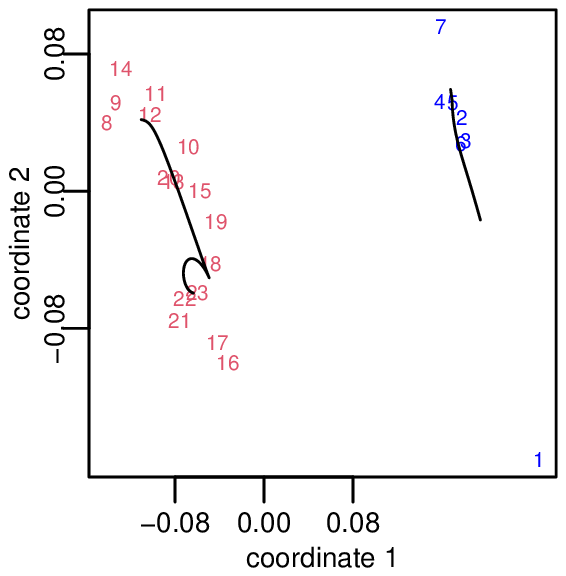}
    \end{subfigure}
     \vspace{0.5cm}
     \caption{
     Regression paths for Austen novels (blue) 
     between the years 1794 to 1815, numbered 1-7 according to chronology of the novels,
     and for Dickens novels (red) between the years 1836 to 1870, numbered 8-23; 
     using (left to right) $d=d_1$ and $d=d_\frac{1}{2}$, with bandwidth (top to bottom) $h=1,\, 2,\, 4$, of which $h=4$ gave the smallest value of the cross-validation criterion \eqref{eqn:CV:criterion}. }\label{fig:NW regression novels}
\end{figure}

We take the node set $V$ as the $m=1000$ most common words across all the novels of Austen and Dickens. 
A pre-processing step for the novels is to normalise
each graph Laplacian, in order to remove the gross effects of different lengths of the novels, by dividing each graph
Laplacian by its own trace, resulting in a trace of 1 for each novel.
  Our key
statistical goal is to investigate the authors' evolving writing styles, by carrying out non-parametric 
regression with a graph Laplacian response on the year $t$ that each novel was written.

%\subsection{Nadaraya-Watson regression of networks on time} 
We apply the Nadaraya-Watson regression to the Jane Austen and Charles Dickens networks separately to predict their writing styles at different times. 
The response is a graph Laplacian and the covariate is time $t$ for each novel, with a separate regression 
for each novelist. 
We compared using the metrics $d_1$ and $d_\frac{1}{2}$. For each author a Nadaraya-Watson estimate was produced for every 6 months within the period the author was writing. We compared different bandwidths, $h$, in the Gaussian kernel.  The results are shown in Figure \ref{fig:NW regression novels} plotted on the first and second principal component space for all the novels.
 
 For both metrics for Dickens when $h=1$ the regression lines are not at all smooth. 
For both metrics with $h=2$ the regression line for Dickens appears to show an 
 initial smooth trend, then a
 turning point around the years 1850 and 1851 (between David Copperfield and Bleak House which are novels 16 and 17 in Figure \ref{fig:NW regression novels}).
 After 1851 there is much less dependence 
 on time. This change in structure is especially evident in the $h=4$ plot for both metrics, which has the smallest value of the cross-validation criterion \eqref{eqn:CV:criterion} out of these choices $h \in \{ 1,2,4 \}$. 
  In the year 1851 Dickens had a tragic year including his wife having a nervous breakdown, his father dying and his youngest child dying \citep{charlesDickens}. We see that the possible turning point is around the same time as these significant events. 
 
 As there are far fewer novels written by Austen it is less obvious if there is any turning point in her writing, however it is clear that Lady Susan (novel 1) is an anomaly, not fitting with the regression curve that does follow Austen's other works more closely.
 Lady Susan is Austen's earliest work, and is a short novella published 54 years after Austen's death.

\section{Discussion}
The two applications presented involve a scalar covariate, but the Nadaraya-Watson estimator is appropriate to more general covariates, e.g. spatial covariates. A further extension would be to adapt the method of kriging, also referred to as Gaussian process prediction. 
Kriging is a geospatial method for prediction at points on a random field  \citep[e.g. see][]{Cressie93}, and \citet{Pigolietal16} considered kriging for manifold-valued data. 
The kriging  predictor of an unknown 
graph Laplacian $\mathbf{L}(\mathbf{x})$ on a random field with known coordinates $\mathbf{x}$ for the dataset $(\lbrace \mathbf{L}_1, \mathbf{x}_1\rbrace,\dots, \lbrace \mathbf{L}_n,\mathbf{x}_n\rbrace)$ is of the form $Z(\mathbf{x})=\sum_{i=1}^n b(\mathbf{x}_i)\mathbf{L}_i$,
where the weights, $b(\mathbf{x}_i)$, are determined by minimizing the mean square prediction error for a given covariance function.   
%The working to find these weights for the spatial graph Laplacians is found in Section \ref{sec:kriging}.

%\subsection{Nadaraya-Watson regression of Euclidean response versus graph Laplacian covariate}\label{sec:NW regression II}
The Nadaraya-Watson estimator can also be applied in a reverse setting where some variable $\mathbf{t}_i$ is dependent on the graph Laplacian $\mathbf{L}_i$, this can be written as $\mathbf{t}_i=\mathbf{t}(\mathbf{L}_i)$. This could be used if, for example, one had the times networks were produced and then wanted to predict the time a new network was produced. In this case the Nadaraya-Watson estimator is a linear combination of known $\mathbf{t}_i$ values, weighted by the graph Laplacian distances, given by
\begin{align}\label{eq: NW II}
\hat{\mathbf{t}}(\mathbf{L})=\frac{\sum_{i=1}^nK_h(d(\mathbf{L}, \mathbf{L}_{i}))\mathbf{t}_i}{\sum_{i=1}^nK_h(d(\mathbf{L}, \mathbf{L}_{i}))},
\end{align}
where $d$ can be any metric between two graph Laplacians. \citet{Severn19} provided an application of this approach using the Gaussian kernel defined in (\ref{eq:kernel function}), predicting the time that a novel was written using the network graph Laplacian as a covariate. 

Other metrics could also be used, for example the Procrustes metric of \citet{Severnetal19}. To solve (\ref{eq:manifold nadaraya}) for the Procrustes metric, the algorithm for obtaining a weighted generalised Procrustes mean given in \citet[Chapter 7]{Drydmard16} can be implemented. 

In Euclidean space there are more general results for the Nadaraya-Watson estimator including weak convergence in $L_p$ norm, rather than 
$p=2$ results that we have used  \citep{devroye1980,spiegelman1980}. More general results also exist, e.g. see \citet{Walk2002}, including
strong consistency. It will be interesting to explore which of these results can be extended to graph Laplacians, although the 
additive properties of the $p=2$ case have been particularly important in our work.

\section*{Acknowledgments}

This work was supported by the Engineering and Physical Sciences Research Council [grant number  EP/T003928/1]. The datasets were derived from the following resources available in the public domain: The Network Repository http://networkrepository.com and CLiC https://clic.bham.ac.uk

\bibliographystyle{apalike}
\bibliography{nprpaper-bib}

\end{document}